\documentclass[11pt]{article}
\usepackage[margin=1in]{geometry}
\usepackage{microtype}%if unwanted, comment out or use option "draft"
\usepackage{amssymb}
\usepackage{amsthm}
\usepackage{xspace}
\usepackage{graphicx}
\usepackage{epsfig}
\usepackage{ifpdf}
\usepackage{url,hyperref}
\usepackage{latexsym}
\usepackage{xspace}
\usepackage{color}
\usepackage{makeidx}
\usepackage{stackrel}
\usepackage{amscd}
\usepackage[all]{xy}
\usepackage{lineno}
\usepackage{algorithmic,algorithm}
\usepackage{accents}
\usepackage{picins}

\usepackage{framed}
\usepackage{pb-diagram}

\long\def\remove#1{}

\newtheorem{theorem}{Theorem}
\newtheorem{proposition}[theorem]{Proposition}

\newtheorem{newremark}{Remark}[section]
\newtheorem{observation}{Observation}[section]

\definecolor{darkblue}{rgb}{0.0, 0.0, 0.8}
\definecolor{darkred}{rgb}{0.8, 0.0, 0.0}
\definecolor{darkgreen}{rgb}{0.0, 0.8, 0.0}

\newcommand{\Mod}[1]{(\mathrm{mod}\ #1)}

\newcommand {\mm}[1] {\ifmmode{#1}\else{\mbox{\(#1\)}}\fi}

\newcommand{\real}     {\mathbb{R}}
\newcommand{\Z}                 {\mathrm {\mathbb{Z}}}

\newcommand{\lev}		{{\mathcal L}}
\newcommand{\reeb}		{{\mathrm Rb}}
\newcommand{\sublev}		{{\mathcal {SL}}}

\newcommand{\R}                {\mathrm {\sf{R}}}

\newcommand{\K}         {{\cal K}}
\newcommand{\F}         {{\cal F}}

\newcommand{\ff}                {\mathbf{\kappa}}

\newcommand{\lst}            {\mm {\rm list}}
\newcommand{\tree}               {T}
\newcommand{\prev}               {\rm prev}
\newcommand{\nxt}                {\rm next}
\newcommand{\tri}                {\rm tri}
\newcommand{\edgt}                {\rm tail}
\newcommand{\edgh}                {\rm head}

\newcommand{\essential}            {\rm prim}
\newcommand{\bool}            {\rm bool}
\newcommand{\tru}            {\rm true}
\newcommand{\fal}            {\rm false}

\newcommand{\barcode}            {\rm barcode}

\newcommand{\type}            {\rm type}
\newcommand{\ftag}            {\rm fintag}
\newcommand{\height}            {\rm height}
\newcommand{\bone}            {\rm b}
\newcommand{\btree}            {\rm parent}

\newcommand{\homo}      {{\sf H}}

\newcommand{\opbarhom}      {\accentset{\circ}{\sf B}}

\newcommand{\clbarhom}      {\bar{\sf B}}

\newcommand{\VecU}{\mathbb{U}}
\newcommand{\VecV}{\mathbb{V}}

\newcommand{\tspace}          {{\mathsf{X}}}

\title{Computing Height Persistence and Homology Generators
	in $\mathbb{R}^3$ Efficiently}

\author{Tamal K. Dey\thanks{Department of Computer Science and Engineering, The Ohio State University. \texttt{tamaldey@cse.ohio-state.edu}}}

\date{}

\begin{document}

\maketitle

%\newpage
%\setcounter{page}{1}
%\linenumbers
\begin{abstract}
	Recently it has been shown that computing the dimension of the
	first homology group $\homo_1(\K)$ of a simplicial $2$-complex $\K$
	embedded linearly in
	$\mathbb{R}^4$ is as hard as computing the rank of a sparse $0-1$
	matrix. This puts a major roadblock to computing persistence and 
	a homology basis (generators) for complexes embedded in $\mathbb{R}^4$
	and beyond in less than quadratic or even near-quadratic time. 
	But, what about dimension three?
	It is known that when $\K$ is a graph or a surface 
	with $n$ simplices linearly embedded in $\mathbb{R}^3$, the persistence for piecewise linear functions
	on $\K$ can be computed in $O(n\log n)$ time 
	and a set of generators of total size $k$
	can be computed in $O(n+k)$ time . However, 
	the question for general simplicial complexes $\K$
	linearly embedded in $\mathbb{R}^3$ is not completely settled.
	No algorithm with a complexity better than that of 
	the matrix multiplication 
	is known for this important case. 
	We show that the persistence for
	{\em height functions} on such complexes, 
	hence called {\em height persistence}, 
	can be computed in $O(n\log n)$ time. 
	This allows us to compute a basis (generators) of $\homo_i(\K)$, $i=1,2$,
	in $O(n\log n+k)$ time where $k$ is the size of the output. This improves significantly the current best bound of $O(n^{\omega})$, $\omega$ being the exponent of matrix multiplication.
	%Also, we observe that a slight modification of our method applied
	%to PL functions on graphs and simplicial complexes linearly embedded
	%in $\mathbb{R}^2$ provides an $O(n)$ algorithm for persistence
	%even when no explicit
	%vertex order is given. 
	We achieve these improved bounds by leveraging 
	recent results on zigzag persistence in computational topology,
	new observations about Reeb graphs,
	and some efficient geometric data structures.
\end{abstract}

\newpage
\setcounter{page}{1}
%\linenumbers
\section{Introduction}
\label{sec:intro}
Topological persistence for a filtration or a piecewise linear
function on a simplicial complex $\K$ is known to be computable
in $O(n^{\omega})$ time~\cite{MMS11} where $n$ is the number of simplices
in $\K$ and $\omega<2.373$ is the exponent of matrix multiplication. The question regarding the  
lower bound on its computation was largely open until 
Edelsbrunner and Parsa~\cite{EP14} showed that computing the rank 
of the first homology group $\homo_1(\K)$ of a simplicial complex
$\K$ linearly embedded in $\mathbb{R}^4$  
is as hard as the rank computation of a sparse $n\times n$  $0$-$1$ matrix.
The current upper bound for matrix rank computation is 
super-quadratic~\cite{CKL13} and lowering it is a well-recognized hard problem.
Consequently, computing the
dimension of the homology groups and hence the topological persistence
for functions on general complexes in better than super-quadratic 
time is difficult, if not impossible. 
But, what about the special cases that are still interesting?
The complexes embedded in three dimensions which arise
in plenty of applications present such cases.

It is easy to see that the Betti numbers $\beta_i$, the rank of the $i$th homology group $\homo_i(\K)$ defined over a finite field
for a simplicial
complex $\K$ linearly embedded in $\mathbb{R}^3$ can be
computed in $O(n)$ time. For this, compute
$\beta_2$  with a walk over the boundaries of the voids,
compute $\beta_0$ as the number of components of $\K$, and
then compute $\beta_1$ from the Euler characteristics of $\K$ obtained as
the alternating sum of the numbers of simplices of each dimension.
Unfortunately, computation of other topological properties such
as persistence and homology generators (basis) for such a complex $\K$ 
is not known to be any easier than that of
matrix multiplication ($O(n^{\omega})$ time). 
In the special case
when $\K$ is a graph or a surface,
the persistence for a PL function or a filtration
on $\K$ can be computed in $O(n\log n)$ time~\cite{AEHW06,DLL10}.
In this paper, we show that when $\K$ is more general, that is,
a simplicial complex linearly embedded in $\mathbb{R}^3$,
the persistence of a height function on it can be computed
in $O(n\log n)$ time. This special type of persistence which
we term as the {\em height persistence} is not as general as the
standard persistence. Nonetheless, it provides an avenue
to compute a set of {\em basis cycles} in $O(n\log n +k)$ time
where $k$ is the total size of the output. Also, the height persistence
provides a window to the topological features of the
domain $\K$, the need for which arises in various applications.

To arrive at our result, we first observe a connection
between the standard sublevel-set persistence~\cite{ELZ02,ZC05}
and the level-set zigzag persistence~\cite{CSM09} from the
recent work in~\cite{BCE12,BD13,CSM09}. Then, with a sweep-plane algorithm that treats
the level sets as planar graphs embedded in a plane, we compute
a {\em barcode graph} in $O(n\log n)$ time. A barcode is extracted from this graph using a slight but important modification of an algorithm in~\cite{AEHW06}. The barcode extracted from this 
graph provides a part of the height persistence. We show that 
the missing piece can be recovered from the Reeb graph which
can be computed again in $O(n\log n)$ time~\cite{Parsa}. 
We make other observations that allow us to extract the
actual basis cycles from both pieces in $O(n\log n +k)$ time
as claimed.
\section{Background}
A zigzag diagram of topological spaces is a sequence 
\begin{equation}
{\cal X}: \tspace_0\leftrightarrow\tspace_1\leftrightarrow\cdots\leftrightarrow\tspace_m
\label{zigzagd}
\end{equation}
where each $\tspace_i$ is a topological space and each bidirectional
arrow `$\leftrightarrow$' is either a forward or a backward continuous
map. Applying the homology functor with coefficient in a field $\ff$,
we obtain a sequence of vector spaces connected by forward or backward
linear maps, also called a zigzag module:
\begin{eqnarray*}
\homo_p({\cal X}): \homo_p(\tspace_0)\leftrightarrow\homo_p(\tspace_1)
\leftrightarrow\cdots\leftrightarrow\homo_p(\tspace_m)
\end{eqnarray*} 

When all vector spaces in $\homo_p(\cal X)$ are finite dimensional,
the Gabriel's theorem in quiver theory~\cite{Gabriel} says that $\homo_p(\cal X)$ is a 
direct sum of a finite number of interval modules which are of
the form
\begin{eqnarray*}
{\cal I}_{[b,d]}: I_1\leftrightarrow I_2\cdots \leftrightarrow I_m
\end{eqnarray*}
where $I_j=\ff$ for $b\leq j\leq d$ and $\mathbf{0}$ otherwise with the
maps $\ff\leftarrow \ff$ and $\ff\rightarrow \ff$ being identities.
The decomposition $\homo_p({\cal X})=\bigoplus_i {{\cal I}}_{[b_i,d_i]}$
provides a {\em barcode} (set of interval modules) for topological persistence when the topological
spaces $\tspace_i$ originate as sublevel or level sets of a 
real-valued function $f:\tspace\rightarrow\real$ defined on a space
$\tspace$. As shown in~\cite{CSM09},
classical persistence~\cite{ELZ02,ZC05}, its extended
version~\cite{CEH09}, and the more general zigzag persistence~\cite{CSM09}
arise as a consequence of choosing variants of the module ${\cal X}$ 
in~\ref{zigzagd} that are derived from $f$.

\subsection{Standard persistence}
Standard persistence~\cite{ELZ02,ZC05} is 
defined by considering 
the sublevel sets of $f$, that is, $\tspace_i$ is $f^{-1}(-\infty,a_i]$ for
some $a_i\in\real$. These values $a_i$ are taken as the critical
values of $f$ so that the barcode captures
the evolution of the homology classes of the sub-level sets
across the critical values of $f$, which  
are defined below precisely.

For an interval $I\subseteq \real$,
let $\tspace_I:=f^{-1}(I)$ denote the interval set.
Following~\cite{BCE12,CSM09}, we assume that $f$ is tame. It means
that it has finitely many homological critical values 
$a_1<a_2<\cdots<a_m$ so that for each open interval 
$I\in \{(-\infty,a_1),(a_1,a_2),\ldots,(a_{m-1},a_m),(a_m,\infty)\}$,
$\tspace_I$ is homeomorphic to a product space ${\mathbb Y}\times I$,
with $f(\mathbb{Y})\in I$. This homeomorphism should extend to
a continuous function 
$\tspace_{\bar I}\rightarrow {\mathbb{Y}}\times {\bar I}$,
with $\bar I$ being the  
closure of $I$ and each interval
set $\tspace_I$ should have finitely generated homology groups. 

It turns out that the description of the interval modules assumes one more
subtle aspect when it comes to describing the  
standard persistence and zigzag persistence in general. 
Specifically, the interval modules can be
{\em open} or {\em closed} at their end points. 
To elucidate this,
consider a set of values $\{s_i\}$ of $f$ 
interleaving with its critical values:
$$
s_0 < a_1 < s_1 <\ldots <a_m<s_m
$$
Assuming $a_0=-\infty$ and $a_{m+1}=\infty$, one can write
the sub-level sets as $\tspace_{[0,r]} :=f^{-1}(-\infty,r]$. 
For standard persistence, we consider the sublevel set diagram
and its corresponding homology module $\homo_p(\sublev(f,\tspace))$ for
dimension $p\geq 0$:
\begin{eqnarray*}
	\sublev(f,\tspace): \tspace_{[0,a_1]}\rightarrow \tspace_{[0,s_1]}
	\rightarrow\tspace_{[0,a_2]}\cdots
	\rightarrow \tspace_{[0,s_m]}\rightarrow\tspace_{[0,a_{m+1}]}\\
	\homo_p(\sublev(f,\tspace)): \homo_p(\tspace_{[0,a_1]})\rightarrow \homo_p(\tspace_{[0,s_1]})
	\rightarrow\homo_p(\tspace_{[0,a_2]})\cdots
	\rightarrow \homo_p(\tspace_{[0,s_m]})\rightarrow\homo_p(\tspace_{[0,a_{m+1}]})
\end{eqnarray*}
The summand interval modules, or the so called {\em bars}, for
this case has the form $[a_i,s_j]$. This means that a $p$-dimensional
homology class
is born at the critical value $a_i$ and it dies at the value $s_j$.
The right end point of $s_j$ is an artifact of our choice of the
intermediate value $s_j\in (a_j,a_{j+1})$. Because of our
assumption that $f$ is tame, homology classes cannot die
in any open interval between the critical values. In fact, they remain
alive in the interval $(a_j,a_{j+1})$ and may die entering the
critical value $a_{j+1}$. To accommodate this fact, we
convert each bar $[a_i,s_j]$ of the standard persistence
to a bar $[a_i,a_{j+1})$ that is open on the right end point.

One can see that there are two types of bars in the standard
persistence, one of the type $[a_i,a_j)$, $j\not=m+1$, which is
bounded (finite) on the right, and
the other of the type $[a_i,\infty)$ which is unbounded (infinite) on the right.
The unbounded bars represent the essential homology
classes since  $\homo_p(\tspace)\cong\bigoplus_i [a_i,\infty)$. 
The work of~\cite{BCE12,BD13,CSM09} implies that both 
types of bars of
the standard persistence can be recovered from those of the
level set zigzag persistence as described next. This observation leads to an efficient
algorithm for computing the standard persistence in $\real^3$. 

\subsection{Level set zigzag}
In level set zigzag persistence, we track the changes in the
homology classes in the level sets $\tspace_{r}=f^{-1}(r)$
instead of the sub-level sets. We need maps connecting
individual level sets, which is achieved 
by including the level sets into the adjacent interval sets.
For this purpose we use the notation $\tspace_i^j:=\tspace_{[s_i,s_j]}$
for the interval set between the two non-critical level sets.
We have a zigzag sequence of interval and
level sets connected by inclusions producing a level set
zigzag diagram:
\begin{eqnarray*}
\lev(f,\tspace): \tspace_0^0\rightarrow \tspace_0^1 \leftarrow \tspace_1^1\rightarrow
\tspace_1^2\cdots\rightarrow \tspace_{m-1}^m\leftarrow \tspace_m^m.
\end{eqnarray*}
Applying the homology functor $\homo_p$ with coefficients
in a field $\ff$, we obtain the 
zigzag persistence module for any dimension $p\geq 0$
\begin{equation}
\homo_p(\lev(f,\tspace)): \homo_p(\tspace_0^0)\rightarrow \homo_p(\tspace_0^1) \leftarrow \homo_p(\tspace_1^1)\rightarrow
\cdots\rightarrow \homo_p(\tspace_{m-1}^m)\leftarrow \homo_p(\tspace_{m}^m).
\label{eqn-zigzag}
\end{equation}
The zigzag persistence of $f$ is given by the 
summand interval modules of $\homo_p(\lev(f,\tspace))$. Each interval module
is of the type $[r,r']$ where $r$ and $r'$ can be
$a_i$ or $s_i$ for some $i\in [0,m+1]$.
Just as in the sub-level set persistence, we identify the 
end points of the interval modules with the critical values \parpic[r]{\includegraphics[height=3.5cm]{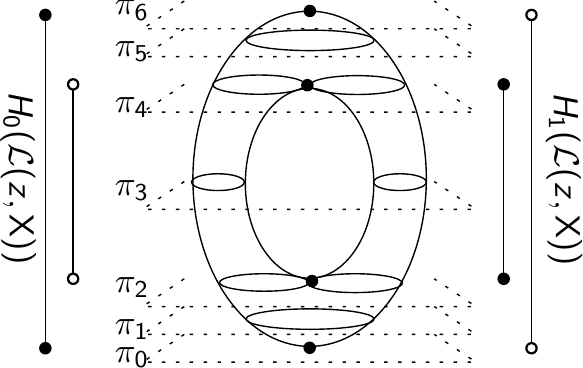}}
that were used to define the
level set zigzag in the first place. 
In keeping with the understanding that even the level set
homology classes
do not change in the open interval sets, we
convert an endpoint $s_i$ to an adjacent critical value
and make the interval module open at that critical value. 
Precisely we modify the interval modules as
(i) $[a_i,a_j]\Leftrightarrow [a_i,a_j]$,
(ii) $[a_i,s_j]\Leftrightarrow [a_i,a_{j+1})$
(iii) $[s_i,a_j]\Leftrightarrow (a_i,a_j]$
(ii) $[s_i,s_j]\Leftrightarrow (a_i,a_{j+1})$.
The intervals in (i)-(iv) are referred as {\em closed-closed}, 
{\em closed-open}, {\em open-closed}, and {\em open-open} bars
respectively. The figure above shows the two bar codes, one for
$\homo_0$ and another for $\homo_1$ for a height function on a torus.
The rightmost picture shows the barcode graph of $\homo_1(\lev(z,{\sf X}))$
which we explain later.

%The mapping of the interval modules $[i,j]$ to the critical values
%gives rise to open and closed intervals--a concept crucial for developing the
%algorithms later. If any of $i$ and $j$, say $i$, in an interval module
%$[i,j]$ is even, it corresponds to the slice $\tspace_i^i$
%representing the open interval set between the crtical values
%$a_{i}$ and $a_{i+1}$. Conversely, if $i$ is odd, it corresponds
%to the slice $\tspace_{i-1}^i$ representing the level set
%at the critical value $a_i$. Therefore, we translate an 
%interval module $[i,j]$ to an interval of critical values
%(i) $[a_i,a_j]$ when $i$ and $j$ are odd, 
%(ii) $[a_i,a_j)$ when 
%$i$ odd and $j$ even, (iii) $(a_i,a_j]$ when $i$ even and $j$ odd,
%and (iv) $(a_i,a_j)$ when $i$ and $j$ are both even. 
%The intervals in (i)-(iv) are referred as {\em closed-closed}, 
%{\em closed-open}, {\em open-closed}, and {\em open-open} bar code
%respectively.

Using the results in~\cite{BD13,CSM09},
we can connect the standard persistence with the level set zigzag 
persistence as follows:

\begin{theorem}
	~
	\begin{enumerate}
		\item $[a_i,a_j)$ is a bar for $\homo_p(\sublev(f,\tspace))$ 
		iff it is so for $\homo_p(\lev(f,\tspace))$,
		\item $[a_i,\infty)$ is a bar for $\homo_p(\sublev(f,\tspace))$
		iff either $[a_i,a_j]$ is a closed-closed bar for
		$\homo_p(\lev(f,\tspace))$ for some $a_j>a_i$, or
		$(a_j,a_i)$ is an open-open bar for 
		$\homo_{p-1}(\lev(f,\tspace))$ for 
		some $a_j<a_i$.
	\end{enumerate}
	\label{main-thm}
\end{theorem}
\begin{proof}
	We know $\homo_p(\sublev(f,\tspace))\cong 
	(\oplus_{i,j}[a_i,a_j)) \bigoplus (\oplus_i [a_i,\infty))$.
	The first summand given by the finite intervals is
	isomorphic to a similar summand $\oplus_{i,j}[a_i,a_j)$ in
	the level set zigzag module $\homo_p(\lev(f,\tspace))$;
	see~\cite{CSM09}(Table 1, Type I). The second
	summand is isomorphic to $\homo_p(\tspace)$,
	which by a result in~\cite{BD13} is isomorphic to 
	$\opbarhom_{p-1}(f,\tspace)\oplus \clbarhom_p(f,\tspace)$ 
	where the open-open interval modules in $\homo_{p-1}(\lev(f,\tspace))$
	generate $\opbarhom_{p-1}(f,\tspace)$ 
	and the closed-closed interval modules in $\homo_p(\lev(f,\tspace))$
	generate $\clbarhom_p(f,\tspace)$.
	Then, the claimed result follows again 
	from~\cite{CSM09}(Table 1, Type III and IV).
\end{proof}

\vspace{0.1in}
\noindent
{\bf Overview and main results.}
Let $\K$ be a simplicial complex consisting of $n$ simplices that are
linearly embedded in $\real^3$. Let $|\K|$ denote the geometric realization arising out of this embedding. 
First, assume that $\K$ is a pure $2$-complex, that is, its highest
dimensional simplices are triangles and all vertices and edges
are faces of at least one triangle. The algorithm for the
case when it has tetrahedra and possibly edges and vertices
that are not faces of triangles  
follows straightforwardly from the case when $\K$ is pure,
and is remarked upon at the end. Another assumption we make for our algorithm is that the coefficient field $\kappa$ of the homology groups is $\Z_2$.

A function $f: |\K|\rightarrow \real$ is called a {\em height} function
if there is an affine transformation $T$ of the coordinate frame
so that $f(x)=z(T(x))$ for all points $x\in |\K|$ with $z$-coordinate 
being $z(x)$.
Without loss of generality, assume that $f$ is indeed the 
$z$-coordinate function
and $z$ is proper, that is, its values on the vertices are distinct. 
The standard topological persistence of $z$ on $|\K|$ is called
the {\em height persistence} which we aim to compute. 
Theorem~\ref{main-thm} says that we can compute the barcode of the
height persistence by computing the same for the
level set zigzag persistence
using the same height function. Precisely, we first compute the barcode
for $\homo_1(\lev(z,|\K|))$ from which we obtain a partial set of bars
for $\homo_1(\sublev(z,|\K|))$ and the complete set of bars for 
$\homo_2(\sublev(z,|\K|))$.
This is achieved by maintaining a level set
data structure and tracking a set of {\em primary} cycles in them
as we sweep through $|\K|$ along increasing $z$. At the same time, we build
a barcode graph that registers the birth, death, split, and merge
of the primary cycles. We show that this can be done 
in $O(n\log n)$ time. The bars of $\homo_1(\lev(z,|\K|))$ are extracted
from this graph again in $O(n\log n)$ time by adapting an algorithm 
of~\cite{AEHW06} to our case after a slight but important modification. According to Theorem~\ref{main-thm},
the closed-open and closed-closed bars of $\homo_1(\lev(f,|\K|))$ constitute
a partial set of bars for $\homo_1(\sublev(z,|\K|))$. The open-open
bars of $\homo_1(\lev(f,|\K|))$, on the other hand,
constitute a complete list of bars for the 
second homology module $\homo_2(\sublev(z,|\K|))$ because the other
summands for $\homo_2(\sublev(f,|\K|))$ are trivial.

The rest of the bars
of $\homo_1(\sublev(z,|\K|))$ which are the open-open bars 
of $\homo_0(\lev(z,|\K|))$ (Theorem~\ref{main-thm}) 
are shown to be captured by the Reeb graph
of $z$ on $|\K|$ which can be computed in $O(n\log n)$ time~\cite{Parsa}. 
We show that the basis cycles for the first and second homology groups
can be computed as part of the level set persistence 
and Reeb graph computations. 

\begin{theorem}
Let $\K$ be a simplicial complex embedded in $\mathbb{R}^3$ with $n$ simplices.
Let $z: |\K|\rightarrow \mathbb{R}$ be a height function defined on it. One can compute the barcode for $\homo_i(\lev(z,|\K|))$ for $i=0,1,2$, in $O(n\log n)$ time where $n$ is the number of simplices in $\K$. Furthermore, a set of basis cycles for $\homo_i(\K)$, $i=0,1,2$, can be computed in time $O(n\log n +k)$ where $k$ is the total size of the output cycles.
\label{main1-thm}
\end{theorem}
 Similar statement holds for standard persistence.
\begin{theorem}
Let $\K$ be a simplicial complex embedded in $\mathbb{R}^3$ with $n$ simplices.
Let $z: |\K|\rightarrow \mathbb{R}$ be a height function defined on it. One can compute the barcode for $\homo_i(\sublev(z,|\K|))$ for $i=0,1,2$, in $O(n\log n)$ time where $n$ is the number of simplices in $\K$.
\label{main2-thm}
\end{theorem}

\section{Level set data structure}
Let $v_1,v_2,\ldots,v_{m}$ be the set of vertices
of $\K$ ordered by increasing $z$-values, that is,
$z(v_j)>z(v_i)$ for $j>i$. 
Consider sweeping $|\K|$ in the increasing order of $z$-values. A level set
$|\K|_r:=z^{-1}(r)$, $r\in \real$, viewed as a graph 
embedded in the plane $\pi_r= \{x\in\mathbb{R}^3\mid z(x)=r\}$,
does not change its adjacency structure in any open
interval $(z(v_i),z(v_{i+1}))$. This structure, however, may change
as the level set sweeps through a vertex of $\K$. Consequently, for
every vertex $v_i\in \K$, 
it suffices to track the changes when the level set
jumps from the intermediate level $s_{i-1}< z(v_i)$ 
to the level $a_i:=z(v_i)$ and
then to the intermediate level $s_i>z(v_i)$ where
$z(v_0)=-\infty$, and $z(v_{m+1})=\infty$.
All three level sets $|\K|_{s_{i-1}}$, $|\K|_{a_i}$,
and $|\K|_{s_i}$ are plane 
graphs embedded linearly in the planes $z=s_{i-1}$, $a_i$,
and $s_i$ respectively. Let $G_{r}=(V_{r},E_{r})$
denote any such generic level set graph at a level $r$, where the vertex set $V_{r}$ is the restrictions of the level set to the edges
of $\K$ and the edge set $E_{r}$ is
the restriction of the level set to the triangles of $\K$.
To avoid confusions, we will say {\em complex edges} and
{\em complex triangles} to refer to the edges and triangles 
of $\K$ respectively.

\vspace{0.1in}
\noindent
{\bf Level set graph and homology basis.}
We need to track a set of cycles representing a homology basis
of $\homo_1(G_{s_{i-1}})$ 
to that of $\homo_1(G_{a_i})$ and then
to that of $\homo_1(G_{s_i})$ as we sweep through the
vertex $v_i$. Consider any such generic level set graph $G_r=(V_r,E_r)$ 
representing $z^{-1}(r)$. 
\piccaption{Primary and secondary cycles.\label{prim-pic}}
\parpic[r]{\includegraphics[width=3.5cm]{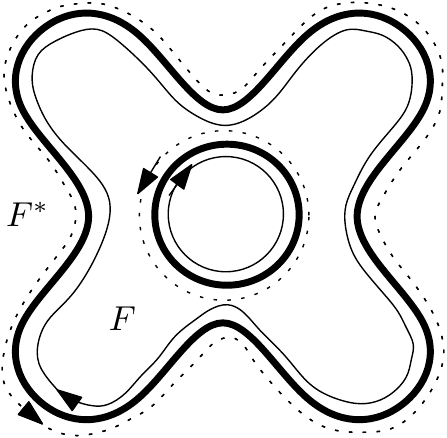}}
The embedding of $G_r$ in the plane $\pi_r$ produces a partition
of $\pi_r$ into 
$2$-dimensional faces, $1$-dimensional edges, and $0$-dimensional vertices.
The faces are 
the connected components 
of $\pi_r\setminus G_r$. 
Let $\F_r$ denote the collection of all $2$-faces in this partition.
A face $F\in\F_r$
has boundary cycle $\partial F$ consisting of possibly multiple components, each being a cycle. 
We orient $F$ by orienting its 
boundary and denote it with $\overrightarrow{F}$.
The orientation is such that $\partial\overrightarrow{F}$ 
has the face on its right. In Figure~\ref{prim-pic}, 
the face $F$ has two boundaries, one around
the outer curve (shown solid) and another around the
inner circle (shown dotted).
The unique face in $\F_r$ that is unbounded plays
a special role and is denoted $F^*$. 
\begin{observation}
	For a bounded face $F\in \F_r$, there is a unique oriented cycle
	$C_{\overrightarrow{F}}\in \partial \overrightarrow{F}$ that
	bounds a bounded face of $\pi_r\setminus C_F$ on its right. By definition,
	the unbounded face $F^*$ has no such $C_{\overrightarrow{F^*}}$.
	In the figure above, $C_{\overrightarrow{F}}$ is the solid
	curve around outer boundary. 
	\label{essen-obs}
\end{observation}

Because of the uniqueness of the cycles $C_{\overrightarrow{F}}$, 
we give them the special name of
{\em primary} cycles. All other cycles are {\em secondary}. 
In Figure~\ref{prim-pic}, the primary cycles are rendered solid and the
secondary ones are rendered dotted.
Recall that the elements of the first homology group $\homo_1$
are classes of cycles denoted $[C]$ for a cycle $C$.
It turns out that the classes of unoriented
primary cycles form a basis for $\homo_1(G_r)$ and thus 
tracking the primary cycles across the levels
become the key to computing the level set zigzag persistence.
\begin{proposition}
	The classes of unoriented cycles $\{[C_F]\,|\, C_{\overrightarrow{F}} \mbox{ is primary}\}$
	form a basis of $\homo_1(G_r)$.
	\label{primary}
\end{proposition}
\begin{proof}
We observe the following facts:
\begin{itemize}
    \item The classes of unoriented primary cycles 
	form a sub-basis of $\homo_1(G_r)$.
	\item $\homo_1(G_r)\cong\tilde{\homo}_0(\pi_r\setminus G_r)$
	where $\tilde{\homo}_0$ denotes the {\em reduced} zero-dimensional 
	homology group.
	\item The faces in $\F_r\setminus F^*$ form a basis of 
	$\tilde{\homo}_0(\pi_r\setminus G_r)$. 
\end{itemize}
	For the first fact, observe that
	the set of such cycles are independent meaning that
	there is no unoriented primary cycle $C_F$ that can be written
	as the sum of other unoriented primary cycles. If it were true,
	let $C_F=C_{F_1}+C_{F_2}+\cdots+C_{F_t}$. Then,
	the boundary of $R=F\cup_{i=1}^t F_i$ is empty. But, that is
	impossible unless $R=\pi_r$. Since $F^*\not\in R$, we
	have $R\not=\pi_r$. 
	
	The second fact follows from Alexander duality because
	$G_r$ is embedded in the plane $\pi_r$. 
	The third fact follows from the
	definition of reduced homology  groups.
	
	Consider a map $\mu$ that sends each
	face $F\in \F_r\setminus F^*$ to its unoriented primary cycle $C_F$.	This map is bijective due to Observation~\ref{essen-obs}. 
	Therefore, by the first and third facts, $\tilde{\homo}_0(\pi_r\setminus G_r)$
	is isomorphic to the summand of $\homo_1(G_r)$ generated by 
	the classes of unoriented primary cycles. 
	Indeed, this summand is $\homo_1(G_r)$ itself
	since  $\tilde{\homo}_0(\pi_r\setminus G_r)$ is isomorphic to 
	$\homo_1(G_r)$ by the second fact.
\end{proof}
The following Proposition complements Proposition~\ref{primary}.
We do not use it, but remark about its connection to Reeb graphs 
at the end.
\begin{proposition}
	The $\homo_0$-classes of unoriented secondary cycles form 
	a basis of $\homo_0(G_r)$.
	\label{secondary}
\end{proposition}

\vspace{0.1in}
\noindent
{\bf Representing level set graphs.}
Proposition~\ref{primary} implies that we can maintain a basis
of $\homo_1(G_r)$ by maintaining the primary cycles alone. However, for
realizing the zigzag maps that connect across the level sets (Eqn.~\ref{eqn-zigzag}), we need a different
basis involving both primary and secondary cycles. For each bounded face $F\in \F_r\setminus F^*$, let $\partial F=C_F+\sum_i C_i$ be the {\em boundary cycle} for the face $F$ which is the $\Z_2$-addition of the primary cycle $C_F$ with the secondary ones in $F$. The next assertion follows from Proposition~\ref{primary} immediately.
\begin{proposition}
The classes $\{[\partial F]\,|\, F\in \F_r\setminus F^*\}$	form a basis of $\homo_1(G_r)$.
\label{primary-secondary}
\end{proposition}

The importance of the boundary cycles in realizing the zigzag maps needed for the persistence module in Eqn.~\ref{eqn-zigzag} is due to the following observation.

\begin{observation}
For every $i\in\{1,\ldots,m-1\}$ and for every boundary cycle $\partial F$ in the intermediate level $s_i$,
there are sum of boundary cycles $\sum\partial F_{i_j}$ and $\sum\partial F_{(i+1)_j}$ at the critical levels $a_i=z(v_i)$ and $a_{i+1}=z(v_{i+1})$ respectively with $a_{i}<s_i<a_{i+1}$ so that the inclusions of $\partial {F_{i_j}}, \partial F$, and $\partial {F_{(i+1)_j}}$ into the interval space $|\K|_{[a_{i},a_{i+1}]}$ induce linear maps at the homology levels given by $[\partial F]\rightarrow [\sum\partial F_{i_j}]$, $[\partial F]\rightarrow [\sum\partial F_{(i+1)_j}]$.
\end{observation}
By Proposition~\ref{primary-secondary} and the above observation, the zigzag maps of the persistence module in Eqn.~\ref{eqn-zigzag} can be tracked if we track the boundary cycles for each face. However, this requires additional bookkeeping for maintaining the primary and secondary cycles of a face together. Instead, we maintain each individual primary and secondary cycle independently being oblivious to their correspondence to a particular face though this information is maintained implicitly. Due to Proposition~\ref{primary}, it becomes sufficient to register the changes in the primary cycles for tracking the boundary cycles.

The primary and secondary cycles change as we sweep over vertices. 
Figure~\ref{levelset} illustrates some of these changes.
A secondary cycle may split into two cycles
one of which is primary and the other is not ($C$ in Fig.), it may 
split into two secondary cycles ($Z$ in Fig.), or two primary cycles
may merge ($D_1$, $D_2$ in Fig.). 
Therefore, we need to maintain all
oriented cycles in $\partial\overrightarrow{F}$, and keep track
of the primary ones among them.

\begin{figure*}[ht!]
	\begin{center}
		\includegraphics[width=0.95\textwidth]{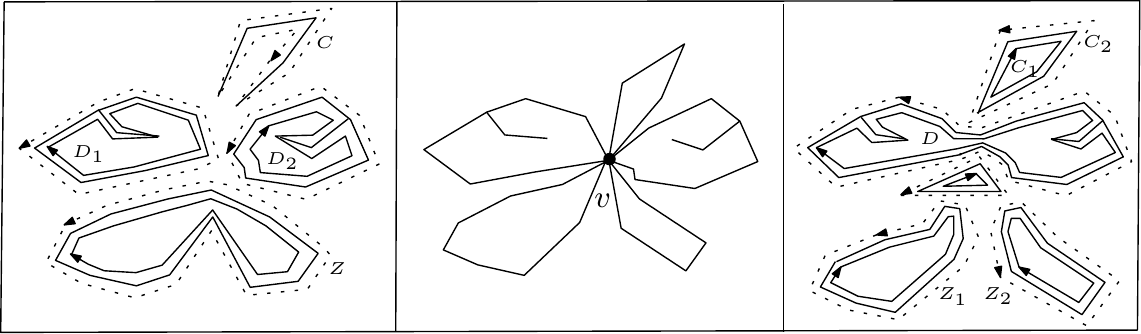}
	\end{center}
	\caption{Level set graph going through changes after sweeping through
		a vertex: $G_{s_{i-1}}$ (left), $G_{a_i}$ (middle), and $G_{s_i}$ (right). The primary
		and secondary cycles 
		%in $\overrightarrow{G}_-$ and
%		$\overrightarrow{G}_+$ 
are indicated with solid and dotted curves
		respectively. Notice how the secondary cycle $C$ on the left %$\overrightarrow{G}_-$}
		got first pinched and then split into one primary 
		cycle $C_1$ and
		another secondary cycle $C_2$ on the right.}
	%	in $\overrightarrow{G}_{+}$}
	\label{levelset}
\end{figure*}
We consider a directed version $\overrightarrow{G_r}=(V_r,\overrightarrow{E_r})$
of $G_r$ where each edge $e\in E_r$ is converted into two 
directed edges in $\overrightarrow{E_r}$ that are oriented oppositely.
The graph $\overrightarrow{G_r}$ is represented with a set
of oriented cycles 
$C(\overrightarrow{G_r})=\cup_{F\in \F_r}\partial \overrightarrow{F}$
that bound the faces in $\F_r$ on right.
These cycles are represented with a sequence of 
directed edges.
\piccaption{Connection rules.\label{connect-pic}}
\parpic[r]{\includegraphics[height=2.5cm]{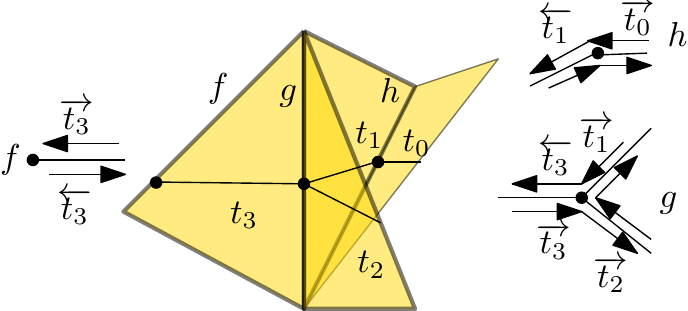}}
A vertex in $V_r$ either lies on a vertex $v\in \K$, or
in the interior of a complex edge $e$ in which case we denote it
as the vertex $e\in V_r$. Any edge in $E_r$ is an intersection of
the level set with a complex triangle $t$, which we also
denote as an edge $t\in E_r$. 
Let $t\in E_r$ be any edge adjoining a vertex $e\in V_r$.
We have two directed copies $\overrightarrow{t}$ and
$\overleftarrow{t}$ of $t$ in $\overrightarrow{G_r}$.
Assume that $\overleftarrow{t}$ is directed away from
$e$ and $\overrightarrow{t}$ is directed toward $e$. 

We follow a {\bf connection rule} for deciding 
the connections among the directed edges around $e\in V_r$ to construct the
cycles in $C(\overrightarrow{G_r})$ as follows. Let $d$ and $d'$
be a pair of 
directed edges, where the head of $d$ is the tail of $d'$. The directed
path $dd'$ locally separates the plane around the meeting point
of $d$ and $d'$. The region to the right of $dd'$ is called its 
{\em right wedge}, and the region to the left is called its {\em left wedge}. 
We have three cases for deciding the connections:
\begin{itemize}
\item $e$ has only one edge $t=t_0$ ($f$ in Figure~\ref{connect-pic}): connect 
$\overrightarrow{t}$ to $\overleftarrow{t}$.
\item $e$ has exactly two edges $t_0$ and $t_1$ ($h$ in Figure~\ref{connect-pic}): connect 
$\overrightarrow{t_0}$ to $\overleftarrow{t_1}$, and connect
$\overleftarrow{t_0}$ to $\overrightarrow{t_1}$.
\item $e$ has three or more edges ($g$ in Figure~\ref{connect-pic}): 
consider a circular order of
all edges $t\in E_r$ adjoining $e\in V_r$. 
Let $t_0,t_1,\ldots,t_k,t_0$
be this circularly ordered edges around $e$. 
For any consecutive pairs of edges $t_i$, $t_{(i+1)\Mod{k}}$,
determine if the right wedge of $\overrightarrow{t_i}\overleftarrow{t}_{(i+1)\Mod{k}}$
contains the edge $t_{(i-1)\Mod{k}}$. If so, connect
$\overrightarrow{t}_{(i+1)\Mod{k}}$ to $\overleftarrow{t_i}$.
If not, connect $\overrightarrow{t_i}$ to
$\overleftarrow{t}_{(i+1)\Mod{k}}$.
\end{itemize}

The choice of our orientations and connections
leads to the following observation:

\begin{observation}
	Let $(d,d')$ be any pair of directed edges in $\overrightarrow{G_r}$.
	They are consecutive directed edges 
	on the oriented boundary 
	of a face $F\in \F_r$ if and only if  
	$d$ connects to $d'$ by the connection rule
	around some vertex $e\in V_r$.
\end{observation}

The observation above relates the directed cycles 
in $C(\overrightarrow{G_r})$ with a local 
connection rule. We exploit this fact to update the cycles 
locally in our algorithm.

\vspace{0.1in}
\noindent
{\bf Cycle trees.}
The directed cycles in $C(\overrightarrow{G_r})$ are represented with
balanced trees that help implementing certain operations on them efficiently. 
We explain this data structure now.

A directed edge $d$ where $d=\overleftarrow{t}$ or 
$d=\overrightarrow{t}$ is represented with
a node $d$ that has three fields; $d\cdot \tri$ points to the complex
triangle $t$, $d\cdot \edgt$ and $d\cdot \edgh$ point to the complex edges
$e_1$ and $e_2$ respectively where $d$ is directed from $e_1$ to $e_2$.
A cycle $C$ of directed edges is  
represented with a balanced tree $\tree_C$ , namely a 2-3 tree~\cite{AHU74} 
where the directed edges of $C$ constitute the leaf nodes of $\tree_C$ with
the constraint that the leaves of any subtree of $\tree_C$ represent
a path (directed) in $C$. The leaves of $\tree_C$ are joined with
a linked list in the order they appear on the directed cycle $C$.
A pointer $d\cdot\nxt$ in a leaf node $d$ implements this link list. 
The node $d$ also maintains another pointer $d\cdot\prev$ to access the 
previous node on the linked list in $O(1)$ time. However, it is important
to keep in mind that it is the $\nxt$ pointers that provide the orientation
of the cycle $C$. Furthermore, the last node in both linked lists 
connected  by
$\nxt$ and $\prev$ pointers respectively is {\em assumed} to connect
to the first one. This creates the necessary circularity without
actually making the list circular. We denote the linked list of
leaves of a tree $\tree$ as $\lst(\tree)$.
The 2-3 trees built on top of the paths support the following
operations.

\vspace{0.1in}
\noindent
{\sc find}($d$): returns the root of the tree $d$ belongs to.\\
{\sc split}($T,d$): \parbox[t]{4.5in}{splits a tree $T$ into two 
	trees $T_1$ and $T_2$
	where $\lst(T_1)$ is the sublist of $\lst(T)$ that
	contains all elements in $\lst(T)$ before $d$,
	and $\lst(T_2)$ is the sublist that contains all elements
	in $\lst(T)$ after and including $d$.}\\ 
{\sc join}($T_1,T_2$): \parbox[t]{4.5in}{takes two trees
	$T_1$ and $T_2$ and produces a single 
	tree $T$ with $\lst(T)$ as the
	concatenation of $\lst(T_1)$ and $\lst(T_2)$ in this order.}\\
{\sc permute}($T,d$): \parbox[t]{4.5in}{makes $d$ the
	first node in the cycle represented with $T$.
	It is implemented by calling 
	{\sc split}($T,d$) that produces $T_1$ and $T_2$, and then
	returning $T:=${\sc join}($T_2$,$T_1$). }\\
%{\sc close}($T$): \parbox[t]{4.5in}{tag $T$ as closed.}\\
{\sc insert}($d,d'$): \parbox[t]{4.5in}{inserts the element $d$ after $d'$ 
	in $\lst(T)$ where $T:=${\sc find}($d'$).}\\
{\sc delete}($d$): \parbox[t]{4.5in}{deletes $d$ from $\lst(T)$ where 
	$T:=${\sc find}($d$).}\\

All of the above operations maintain the trees well balanced allowing traversal of a path from a leaf to the root in $O(\log n)$ time where
$n$ is the total number of elements in the lists of the trees involved. This in turn allows each of these operations to be carried out in $O(\log n)$ time.
Using these basic operations, we implement two key operations,
splitting and merging of cycles.

\vspace{0.1in}
\noindent
{\sc splitcycle}($T,d,d'$): This splits
a directed cycle into two. A cycle may get first pinched and then
splits into more cycles as we sweep through a vertex. This operation
is designed to implement this event. 
Given a tree $T$, it returns two trees $T_1$ and $T_2$
where $\lst(T_1)$ represents the path from $d$ to $d'$ in the directed
cycle given by $\lst(T)$,
and $\lst(T_2)$ represents the path from $d'\cdot \nxt$ to
$d\cdot \prev$ in the same cycle. See Figure~\ref{nullity}, bottom row.

It is implemented as follows: Let $T:=${\sc permute}($T,d$). 
Call {\sc split}($T,d'\cdot\nxt$) which returns two trees $T_1$
and $T_2$ as required. 

\vspace{0.1in}
\noindent
{\sc mergecycle}($d,d'$): This merges the two cycles that 
$d$ and $d'$ belong to.
The new cycle has $d'$ after $d$ and $d\cdot \nxt$ after
$d'\cdot \prev$. This is implemented as follows: 
Let $T_1:=${\sc find}($d$) and $T_2:=${\sc find}($d'$).
Let $T_1:=$ {\sc permute}($T_1,d\cdot\nxt$) and 
$T_2:=${\sc permute}($T_2,d'$).
Then, return $T:=$ {\sc join}($T_1,T_2$).

\begin{figure*}[ht!]
	\begin{center}
		\includegraphics[width=0.95\textwidth]{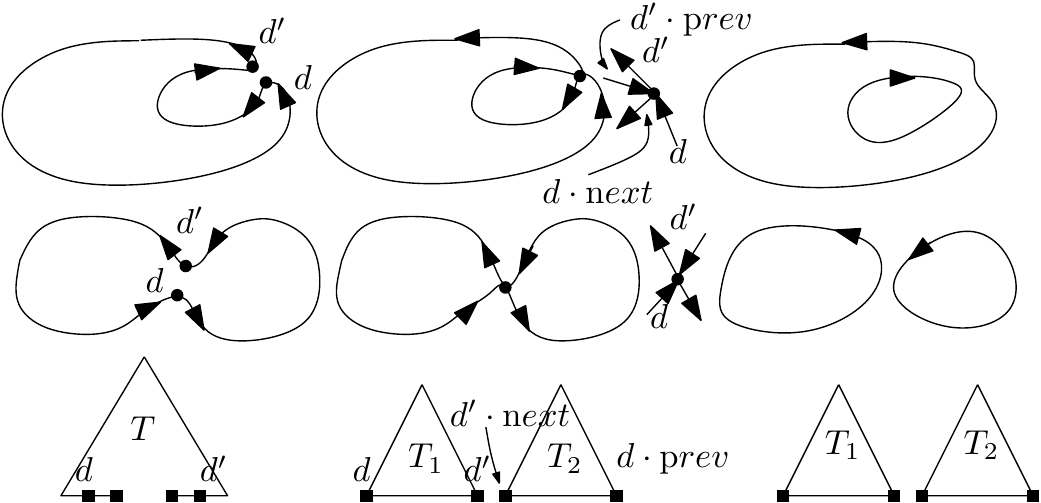}
	\end{center}
	\caption{Secondary cycle splitting: split at the top generates
		one primary and another secondary cycle; split at the bottom generates
		two secondary cycles.}
	\label{nullity}
\end{figure*}

\section{Updating level sets}
Now we describe how we update the graph 
$\overrightarrow{G}_{s_{i-1}}$ to 
$\overrightarrow{G}_{a_i}$ 
and then to $\overrightarrow{G}_{s_i}$.
As we sweep through $v_i$, only the cycles in these graphs
containing a vertex on a complex edge with $v_i$ as an endpoint may change combinatorially.
We only update the cycles for combinatorial changes 
to make sure that the
combinatorics of the level set graphs are maintained correctly though their
geometry is updated only when needed to infer the 
correct adjacencies. This allows us to inspect only $O(n_{v_i})$ simplices
where $n_{v_i}$ is the number of simplices adjoining $v_i$ in $\K$. Summing over
all vertices, this provides an $O(n)$ bound which gets multiplied with
the $O(\log n)$ complexity for the tree
operations that we perform for each such simplex. Also, local
circular sorting of $O(n_{v_i})$ edges around each vertex $v_i$ and complex edges connected to it accounts for $O(n\log n)$ time in total.

\vspace{0.1in}
\noindent
{\bf Primary cycle detection.}
The cycles in $\overrightarrow{G}_r$ that change combinatorially may 
experience splitting, merging, edge contraction, edge expansion, or 
a combination of such events. Specifically, during splitting and merging,
new cycles are generated which need to be characterized as primary or not.
Figure~\ref{nullity} illustrates two cases of a secondary cycle
splitting. Two similar cases arise for the primary cycle splitting. 
For merging also we have four cases mirroring the splitting case.
It turns out that we can determine if the new cycles are primary or not by the orientations of the edges around the `pinching' vertex if we know the type (primary or not) of the original cycles.
We explain this for the case of splitting.

Let $C$ be a cycle in $C(\overrightarrow{G}_{s_{i-1}})$ which
splits at $v_i$. Let $d$ and $d'$ be any two non-consecutive
directed edges in $C$
that meet at $v_i$ in $\overrightarrow{G}_{a_i}$. 
Assume that we know that $C$ is secondary. 
The case when $C$ is primary is similar.
We need to distinguish the case  when one of the two new cycles
nests inside the other. This can be checked in $O(1)$ time by determining
if the right wedge of $d(d\cdot\nxt)$ contains $d'$ or not.
If not, both new cycles remain secondary.
Otherwise, we have a nesting, and exactly one of the two new cycles
becomes primary. We can determine again which of the two becomes
primary in $O(1)$ time. For this consider a ray with tail at $v$
and entering the left wedge of $d(d\cdot\nxt)$. If this
ray enters the left wedge of $dd'$, we declare the new cycle
containing $d$ and $d'$ to be secondary and the other cycle
containing $d'\cdot\prev$ and $d\cdot\nxt$ to be primary.
If the ray enters the right wedge, we flip the assignment for the type of the two new cycles.

With these $O(1)$ local checks,
we design the two routines below that decide the type of
the new cycle(s) in both the splitting and merging cases assuming that we know if
the input cycle(s) are primary or not.

\vspace{0.1in}
\noindent
{\sc splitPrim}($\bool$,$d_1$,$d_2$): This routine assumes that
$\bool$ indicates if the cycle
to be split which contains $d_1$ and $d_2$ is primary (true or false), 
and returns a pair $(\bool_1,\bool_2)$
of booleans where $\bool_i$ is true if and only if the new cycle
containing $d_i$ is primary.

\vspace{0.1in}
\noindent
{\sc mrgPrim}($\bool_1$,$\bool_2$,$d_1$,$d_2$): This routine
assumes that the input boolean variables $\bool_i$ indicates if the cycle containing $d_i$ is primary, and returns
a boolean variable $\bool$ which is true if and only if the
new merged cycle is primary.

\vspace{0.1in}
Now we describe the actual updates of the graphs when the
sweep goes through a vertex $v_i\in \K$. For convenience,
we designate a complex triangle as {\em top, middle}, or {\em bottom} if it 
has $v_i$ as the lowest, middle, or highest vertex respectively
w.r.t. the height $z$.
Similarly, a complex edge is called {\em top, or bottom} if
it has $v_i$ as the lowest or highest vertex respectively.
As we continue with the sweep, we keep on recording the birth, death,
splitting and merging of primary cycles by creating a barcode
graph. {\it Current} primary cycles are represented by {\it current} edges in the barcdoe graph whose one endpoint is already determined, but the other one is yet to be determined. The nodes in the barcode graph are created when a primary cycle is born, dies, splits, or merges with another cycle. It is important to note that the nodes of the barcode graph are created only at the intermediate levels $s_i$. Each tree $T$ maintains  
a pointer $T\cdot\barcode$ that
points to a current edge in the barcode graph if its cycle is primary. Otherwise, this pointer is assumed to be a
null pointer. 
%They are described
%in Appendix~\ref{app:barcode}.
Additionally, we assume that there is a boolean field $T\cdot\essential$
which is set true if and only if $T$ represents a primary cycle. The barcode graph at level $r$ is denoted $B(G_r)$. 
As we move from level $r$ to the next level $r'$, we keep
updating this barcode graph by recording the birth, death,
splitting and merging of primary cycles and still denote it as $B(G_{r})$ till we finish processing level $r'$ at which point we denote it as $B(G_r')$.

\vspace{0.1in} 
\noindent
{\bf Updating $\overrightarrow{G}_{s_{i-1}}$ to $\overrightarrow{G}_{a_i}$.}
The combinatorics of $G_{s_{i-1}}$ change only by the edges $t\in E_{s_{i-1}}$ where
$t$ is a bottom or middle triangle. If the edge $t$ has both vertices
on bottom complex edges, then $t$ is contracted to $v_i$ in $G_{a_i}$. Otherwise,
the edge $t$ remains in $G_{a_i}$, but its adjacency at the vertex which 
becomes $v_i$ in $G_{a_i}$ changes. Also, in both cases classes in
$\homo_1(\lev(z,|\K|))$ may die or be born. We perform the combinatorial changes
and detect the birth and deaths of homology classes
as follows:

\vspace{0.1in}
\noindent
{\it Contracting edges:}
%Every edge $t\in E_-$ where the complex triangle $t$ has $v$ as the
%top vertex contracts to the single point of $v$ in the level set graph
%$G_v=\K_{z(v)}$. We first carry out these contractions to update $G_-$
%to $G_v=(V_v,E_v)$.
When we contract edges, a cycle may simply contract and nothing else happens. 
But, we may also detect that a primary cycle of three edges
is collapsed to two directed edges corresponding to a single
undirected edge. This indicates a death of a class 
in $\homo_1(\lev(z,|\K|))$ which occurs entering 
the level $a_i$ but not exactly at $a_i$. So, we operate as follows.

Let $t$ be any bottom complex triangle for $v_i$, and let
$d_1$ and $d_2$ be two directed edges associated with $t$.
Let $T_1:=${\sc find}($d_1$) and $T_2:=${\sc find}($d_2$).
For $i=1,2$, we call {\sc delete}($d_i$). If $T_i\cdot \essential=\tru$
and $T_i$ has two leaves, 
we terminate the current edge pointed by $T_i\cdot \barcode$ with a closed node at level $s_i$ in $B(G_{s_{i-1}})$ and remove $T_i$ completely. The closed node indicates that the cycle dying entering the level $a_i$ is still alive at the level $s_{i-1}$.

\vspace{0.1in}
\noindent
{\it Cycle updates:}
The edges of a cycle in $C(\overrightarrow{G}_{s_{i-1}})$ can come together
at $v_i$ to create new cycles. After the edge contractions, the only edges that we need to
update for possible combinatorial changes correspond to middle
complex triangles. 
Let $t$ be such a triangle
and let $g$ and $h$ be its edges that are top and bottom edges for
$v_i$ respectively. For each directed edge $d$ with 
$d\cdot \tri=t$, we update $d\cdot \edgt=g$ 
or $d\cdot\edgh=g$ if 
originally we had $d\cdot \edgh=h$ or $d\cdot\edgt=h$ 
respectively.  
%Observe that this is consistent with our assumption that a 
%directed edge $d$ should have vertices on complex edges 
%at or above its bottom vertex. 

\begin{figure*}[ht!]
	\begin{center}
		\includegraphics[width=0.95\textwidth]{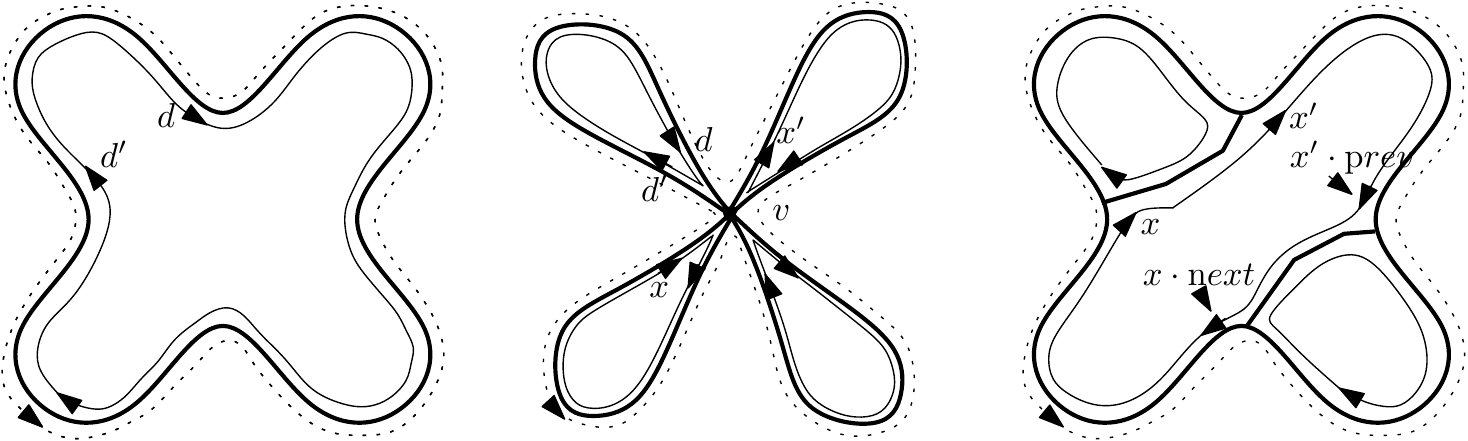}
	\end{center}
	\caption{A primary cycle (left)
	%in $\overrightarrow{G}_-$ (left) 
		splits into four primary cycles (middle),
%		in $\overrightarrow{G}_v$ (middle), 
		then two cycles merge into one
	%	in $\overrightarrow{G}_+$ 
		where the other two only expand (right) when we sweep through vertex $v_i$.}
	\label{split}
\end{figure*}

Next, we update the cycles that may combinatorially
change due to splitting or merging at $v_i$, and also record 
new births as a result.
We consider every directed edge
$d$ so that the triangle $t=d\cdot\tri$ 
is a middle triangle and determine a circular order of their
undirected versions around $v_i$. For every such
directed edge $d$, we determine its pair directed edge
$d'$ using the {\em connection rule} that we described before.
Observe that plane embedding of the level set graph is used here.
Actually, the lack of such canonical ordering of edges around a vertex for level set graphs becomes the roadblock for extending
this algorithm to persistence of functions that are not heights.
Let $T:=${\sc find}($d$) and $T':=${\sc find}($d'$).
We have two cases: the splitting case when $T=T'$ (see
$d$ and $d'$ in Figure~\ref{split}) and the
merging case when $T\not = T'$.

\vspace{0.1in}
\noindent
Splitting Case, $T=T'$: If $d'=d\cdot\nxt$, the cycle 
$C$ containing $d$ and $d'$ and represented by $T$ 
does not change and we do nothing. 
Otherwise, the cycle $C$ splits into two new cycles whose
type needs to be determined.
So, we call {\sc splitPrim}($T\cdot\essential$,$d\cdot\nxt$,$d'$)
which returns a pair of boolean values $(\bool_1,\bool_2)$
indicating if the two new cycles are primary or not.
We split $T$ to create the representations of the two new cycles.
But, this operation destroys $T$ whose
type (primary or not) and barcode pointer are needed for assigning
the same for the two new trees. So, we save 
$b:=T\cdot \barcode$ and $s:=T\cdot\essential$ first, and 
call {\sc splitcycle}($T$,$d$,$d'$) which
returns two trees $T_1$ and $T_2$ representing the
two cycles. Geometric constraints
allow only the following two cases:

\vspace{0.1in}
\noindent
Case(i): $(s,\bool_1,\bool_2)=(\fal,\tru,\fal) \mbox { or }
(\fal,\fal,\tru)$:
A new primary cycle is born at the level $z(v_i)$. 
This is an open-ended birth at the level $s_{i-1}$ because the cycle exists at the level $z(v_i)$ but not at the level $s_{i-1}$.
If $\bool_i=\tru$, we set 
$T_i\cdot\essential:=\tru$
and $T_{i \Mod 2 +1}\cdot\barcode:=null$, 
$T_{i \Mod 2 +1}\cdot\essential:=\fal$.
Then, we set $T_i\cdot\barcode:=b$ where $b$ is a current edge created with an open end at level $s_{i-1}$ in $B(G_{s_{i-1}})$.\\ 
Case(ii): $(s,\bool_1,\bool_2)=(\tru,\tru,\tru)$: A new primary cycle
is born at the level $z(v_i)$ due to a split of the cycle
represented by the saved pointer $b:=t\cdot \barcode$.  
We set $T_i\cdot\essential:=\tru$, 
for $i=1,2$, and call {\sc splitbar}($b$,$v_i$) which splits $b$ at level $s_{i-1}$ and returns two
current edge pointers $\bone_1$ and $\bone_2$.
We set $T_1\cdot\barcode:=\bone_1$ and
$T_2\cdot\barcode:=\bone_2$. 

%\begin{figure*}[ht!]
%\begin{center}
%        \includegraphics[width=0.95\textwidth]{split-tree.pdf}
%\end{center}
%\caption{Merging of the two cycles containing
%directed edges $d$ and $d'$ as shown in Figure~\ref{split} (middle).}
%\label{split-tree}
%\end{figure*}

\vspace{0.1in}
\noindent
Merging Case, $T\not=T'$: two cycles $C_1$ and $C_2$ 
represented by $T_1:=T$ and $T_2:=T'$
respectively merge to become one. As before, we first
store aside the type of $C_1$ and $C_2$ and associated
current edge pointers by
setting $s_i:=T_i\cdot\essential$ and $b_i:=T_i\cdot\barcode$
for $i=1,2$.
Next, we call {\sc mergecycle}($d$,$d'$) which
merges the two cycles containing $d$ and $d'$ and returns a tree
$T_3$ representing this new cycle, say $C_3$. A call to 
{\sc mrgPrim}($s_1$,$s_2$,$d$,$d'$)
returns a boolean variable $\bool$ which is true if 
and only if $C_3$ is primary. 
Again, we have only the following two cases. 

\vspace{0.1in}
\noindent
Case(i): $(s_1,s_2,\bool)=(\tru,\fal,\tru) \mbox { or }
(\fal,\tru,\tru)$.
In this case no primary cycle dies, but the new
cycle remains primary. So, no 
current edge is terminated and the current edge associated
to the primary cycle among $C_1$ and $C_2$ is continued
by $C_3$. If $s_i=\tru$, we set $T_3\cdot\barcode:=b_i$,
$T_3\cdot\essential:=\tru$.\\
Case(ii):$(s_1,s_2,\bool)=(\fal,\fal,\fal)$: No primary cycle
dies and the new cycle is also not primary. 
We set $T_3\cdot\barcode:=null$ and $T_3\cdot\essential:=\fal$.

\vspace{0.1in}
\noindent
{\bf Updating $\overrightarrow{G}_{a_i}$ to $\overrightarrow{G}_{s_i}$.}
To update $\overrightarrow{G}_{a_i}$ to $\overrightarrow{G}_{s_i}$,
we need to create directed edges corresponding to top triangles,
that is, the complex triangles with $v_i$ as the bottom vertex.
These new edges
%\piccaption{Merging lists.\label{merge-pic}}
%\parpic[r]{\includegraphics[height=3.5cm]{split-tree.pdf}}
\begin{figure}[ht!]
\centering{\includegraphics[height=3.5cm]{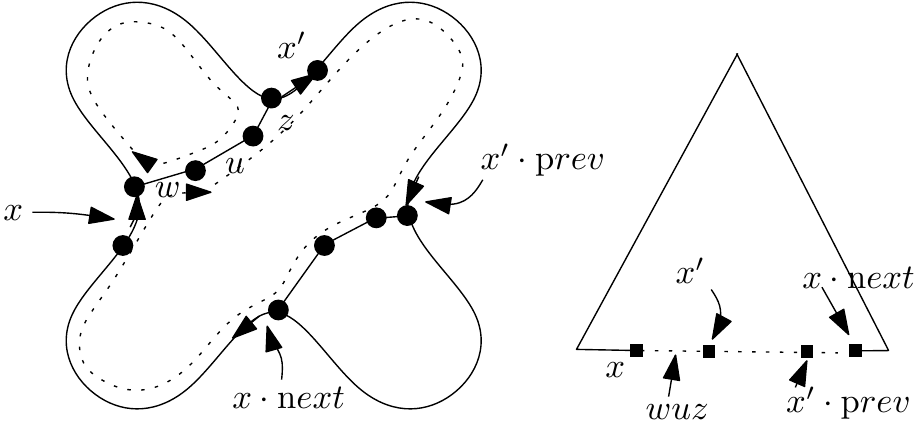}}
\caption{Merging lists.}
\label{merge-pic}
\end{figure}
change the combinatorics of $\overrightarrow{G}_{a_i}$
in four ways: they may (i) expand the existing cycles without
creating or destroying any primary
class, (ii) create
a new cycle giving 
birth to a new class, (iii) 
split a cycle pinched at $v_i$ (it turns out that no 
new class is born in this case), (iv) merge two cycles meeting at $v_i$; in this case, a primary cycle dies. 
Two cycles containing directed edges
$x$ and $x'$ in $\overrightarrow{G}_{a_i}$ in Figure~\ref{split} 
get merged into one cycle in $\overrightarrow{G}_{s_i}$. 
The details of the merge is shown in Figure~\ref{merge-pic}. 
It is preceded by an insertion
of a sequence of edges 
$w,u,z$
that connect $x$ and $x'$.
Similarly, another sequence connects $x'\cdot\prev$ and $x\cdot\nxt$.
%The details of the procedures are described in Appendix~\ref{app:merge}.

{\em Expanding cycles:} We iterate over all directed edges corresponding
to the middle triangles. Let $d$ be any such 
directed edge in $\overrightarrow{G}_{a_i}$. 
The directed edge $d$ belongs to a unique cycle
$C$ in the directed graph $\overrightarrow{G}_{s_i}$. Starting from $d$,
we aim to create the missing edges in $C$.
For this, we create a routine {\sc nextLink}($d$) that takes a directed edge 
$d$ and creates all missing directed edges in $C$ 
that lie between $d$ and the next directed edge $d'$ with
$d'\cdot\tri$ being a middle triangle. 

\noindent
{\sc nextLink}($d$): Consider the complex edge $e=d\cdot \edgh$ and the 
circular order of directed edges of $\overrightarrow{G}_{s_i}$ around $e$. 
In $O(1)$ time we determine the adjacent directed edge $d'$ 
using the connection rule described before.
If $d'$ does not exist already, we create a directed edge for
$d'$ and insert it into
the tree containing $d$ by calling {\sc insert}($d'$,$d$).
Replacing the role of $d$ with $d'$, we continue.
If $d'\cdot\tri$ is a middle triangle, we stop and return $d'$. 

To complete updating $C$ containing the directed edge $d$, we call
{\sc nextlink}($d$) which returns, say $d'$. Let $T:=${\sc find}($d$)
and $T':=${\sc find}($d'$). We have two cases:

\vspace{0.1in}
\noindent
Splitting Case, $T=T'$: If $d'=d\cdot\nxt$ in the graph $\overrightarrow{G}_{a_i}$, then this is a mere expansion of a cycle and we do not make any updates in the barcode graph. Otherwise, we do the same as in the case for updating
$\overrightarrow{G}_{s_{i-1}}$ to $\overrightarrow{G}_{a_i}$ except that the
subcases become different.

\vspace{0.1in}
\noindent
Case(i):  
$(s,\bool_1,\bool_2)=(\tru,\tru,\fal) \mbox { or }
(\tru,\fal,\tru)$.
In this case no new primary cycle is born. So, no new
current edge is created. If $\bool_i=\tru$, we set $T_i\cdot\barcode:=T\cdot\barcode$,
$T_i\cdot\essential:=\tru$
and $T_{{i \Mod 2}+1}\cdot\barcode:=null$, 
$T_{{i \Mod 2}+1}\cdot\essential:=\fal$. 

\vspace{0.1in}
\noindent
Case(ii):$(s,\bool_1,\bool_2)=(\fal,\fal,\fal)$: No primary cycle
is born. We set $T_i\cdot\barcode:=\fal$ and $T_i\cdot\essential:=\fal$.

\vspace{0.1in}
\noindent
Merging Case, $T\not=T'$: Let $C_i$ be represented by $T_i$ where
$T_1:=T$ and $T_2:=T'$. Again, we do the same as in the case of merging
while going from $\overrightarrow{G}_{s_{i-1}}$ to $\overrightarrow{G}_{a_i}$.
The subcases become:

\vspace{0.1in}
\noindent
Case(i): $(s_1,s_2,\bool)=(\tru,\tru,\tru)$: 
Two primary cycles merge to become one. Here one primary class
dies, but we do not know which one. So, we record the merging only.
We join the current edges pointed by $b_1$ and $b_2$ at a node at level $s_i$ and start a new current edge pointed by $\bone$ from that node.
We set $T_3\cdot\barcode:=\bone$, $T_3\cdot\essential:=\tru$.

\vspace{0.1in}
\noindent
Case(ii): 
$(s_1, s_2,\bool)=(\fal,\tru,\fal) \mbox { or }
(\tru,\fal,\fal)$:
A primary cycle dies. So, if $s_i=\tru$ for $i=1$ or $2$, we terminate the current edge
pointed by $b_i$ at level $s_i$ with an open end.  Then,
we set $T_3\cdot\barcode:=null$ and $T_3\cdot\essential:=\fal$.

\vspace{0.1in}
\noindent
{\em New cycles:} Some cycles in $\overrightarrow{G}_{s_i}$ may not arise
from the updates of the old cycles. All of their edges come from the
top triangles that have the vertex $v_i$ as the bottom vertex. These cycles
may introduce new current edges with closed birth at level $s_i$. To create these cycles,
we iterate over all top triangles for which at least one of the two directed edges has not been created yet.
Let $t$ be such a triangle where the directed edge from the complex
edge $e$ to $e'$ has not yet been created. We create the directed
edge $d$ with $d\cdot\tri=t$, $d\cdot\edgt=e$, and $d\cdot\edgh=e'$
and initialize a tree $T$ with it. To complete the cycle $C$ that
$d$ belongs to, we call {\sc nextLink}($d$) which returns
after completing the tree $T$. We check if the new cycle $C$
containing $d$ is primary or not by checking if it contains the
point at infinity. This can be done in $O(n_v)$ time in total for
all such new cycles. If $C$ is primary, 
a current edge $b$ begins with a closed edge end at level $s_i$ in the barcode graph.
So, we set $T\cdot\barcode:=b$ and $T\cdot\essential:=\tru$.
Otherwise, set $T\cdot\barcode=null$ and $T\cdot\essential=\fal$. 

\section{Barcode graph}
After processing the last vertex $v_m$ of $\K$  in the sorted order $v_1,v_2,\ldots,v_m$, we obtain the barcode graph $R=B(G_{s_{m}})$.
It has nodes for the intermediate levels between critical levels of $\K$ 
(levels of vertices of $\K$). 
 Now we proceed to justify why the bars extracted from a modified $R$ are indeed the bars for $\homo_1(\lev(z,|\K|))$.

By considering $R$ as a graph linearly embedded in $\mathbb{R}^3$, we can consider its level set zigzag module with height function $z$. Its vertices have values $s_i$, $i=1,\dots,m-1$ that are linearly interpolated over the edges. Notice that the critical values of $z$ on $R$ are $s_0<s_1<\ldots<s_m$ whereas the same on $\K$ are $a_1<\ldots<a_m$. Writing the interval set $R_{[a_i,a_j]}$ as $R^j_i$, we get the level set zigzag module:
\begin{equation}
\homo_0(\lev(z,R)): \homo_0(R_1^1)\rightarrow \homo_0(R_1^2) \leftarrow \homo_0(R_2^2)\rightarrow
\cdots\rightarrow \homo_0(R_{m-1}^m)\leftarrow \homo_0(R_{m}^m).
\end{equation}
Consider the level set zigzag module $\homo_1(\lev(z,|\K|))$ by putting $p=1$ in (\ref{eqn-zigzag}). Here, the interval sets are $\K_i^j={|\K|_{[s_i,s_j]}}$ (notice the shift in interval sets). We have the level set zigzag module:
\begin{equation}
\homo_1(\lev(z,|\K|)): \homo_1(\K_0^0)\rightarrow \homo_1(\K_0^1) \leftarrow \homo_1(\K_1^1)\rightarrow
\cdots\rightarrow \homo_1(\K_{m-1}^m)\leftarrow \homo_1(\K_{m}^m).
\end{equation}

\vspace{0.1in}
\noindent
{\it Augmeting $R$ with threading:} 
For Proposition~\ref{one-zero-prop} below to be true, we need the homology group of every interval set in the above two modules to be identified with the homology group of the level set at the intermediate value of the vertex. That is, we want $\homo_0(R_i^{i+1})=\homo_0(z^{-1}(s_i))$ and $\homo_1(\K_i^{i+1})=\homo_1(z^{-1}(a_{i+1}))$.
This condition is satisfied for the module for $\K$ because the homology groups for $\K$ do not change except at the critical values $a_i$, $i=1,\ldots, m$. However, this is not true for $R$ because of the open ends of some of the edges. For example, consider a single edge with an open end at the value $s_i$. The homology group $\homo_0(R_i^{i+1})$ in this case has rank $1$ whereas $\homo_0(z^{-1}(s_i))$ has rank $0$ because of the open end node. To remedy this, we consider the reduced homology group $\tilde{\homo}_0(\cdot)$ for $R$ and augment $R$ with an added `thread'. In the above example, if we add a thread with motonic values that attaches to the open end, and then consider the reduced homology group, we get that $\tilde{\homo}_0(R_i^{i+1})\cong \tilde{\homo}_0(z^{-1}(s_i))=0$. Extending this idea, we augment the graph $R$ by adding a `dummy' thread that runs with monotone values in the range $(-\infty,\infty)$ while attaching to every open node of $R$ at every level. See Figure~\ref{thread}. The `dummy' component represented by the thread splits and merge at the `open' degree-1 vertices. Call a degree-1 vertex {\em upward} ($u$ in Figure~\ref{thread}, also see Figure~\ref{barcode-graph}(b)) or {\em downward} ($v$ in Figure~\ref{thread}, also see Figure~\ref{barcode-graph}(b)) if it is connected to a vertex with smaller or larger value respectively. A real component joins with the dummy one at an upward vertex and splits from a dummy component at a downward vertex. The `thread' representing the dummy component turns the open ends into merge or split vertices. 
\piccaption{Threading.\label{thread}}
\parpic[r]{\includegraphics[height=3.5cm]{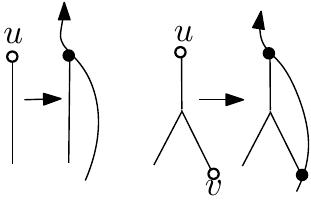}}

The nodes in $R$ at an intermediate level $s_i$ are not supposed to have any edge among themselves. However, because of our construction of $R$, we may have artificial edges between nodes at the same level. See Figure~\ref{barcode-graph}(b), (c). We modify $R$ simply by contracting any such edge (Figure~\ref{barcode-graph}(d)). This operation, carried out in $O(n)$ time, brings all nodes at a fixed level in a connected component formed by artificial edges to a single node at that level. Let $R$ still denote the resulting barcode graph.

\begin{proposition}
$\homo_1(\lev(z,|\K|))\cong \tilde\homo_0(\lev(z,R))$.
\label{one-zero-prop}
\end{proposition}
\begin{proof}
Because of our construction and the threading, the homology group (vector space) of an interval set bracketing a single vertex either in $R$ or in $\K$ is isomorphic to that of the level set at the value of the vertex, that is,  $\tilde \homo_0(R_i^{i+1})=\tilde\homo_0(z^{-1}(s_i))$ and $\homo_1(\K_i^{i+1})=\homo_1(z^{-1}(a_{i+1}))$.
Writing the vector spaces for $R$ as $\tilde\homo_0(z^{-1}(v))=\VecU_{v}$ and the vector spaces for $\K$ as $\homo_1(z^{-1}(v))=\VecV_{v}$, and observing that $\VecU_{v}\cong \VecV_{v}$ for $v=a_i \mbox{ or } s_i$, we obtain the following diagram:
$$
	\xymatrix@C=3ex
	{
		\homo_1(\lev(z,\K)):
		&\VecV_{s_0} \ar[r] \ar@{=}[d]
		& \VecV_{a_1}  \ar@{=}[d]
		& \VecV_{s_1} \ar@{=}[d]\ar[l]
		&\cdots \ar[r] 
		%& \homo_0(\K_1^2) \cdots\ar@{->}[r] \ar@{=}[d]
		& \VecV_{a_m}\ar@{=}[d]
		& \VecV_{s_m} \ar[l]\ar@{=}[d]
		&
		\\
		\tilde\homo_0(\lev(z,R)):
		&\VecU_{s_0} 
		& \VecU_{a_1}\ar[l] \ar[r]
		& \VecU_{s_1} 
		& \cdots
		%& \homo_0(\R_1^2) \cdots\ar@{->}[r]
		& \VecU_{a_m}\ar[l]\ar[r]
		& \VecU_{s_m}
		&
	}
	$$
	To make the above diagram commute, we reverse the arrows for one of the modules, say $\tilde\homo_0(\lev(z,R))$ by considering the dual module on the dual vector spaces:
	$$
	\xymatrix@C=3ex
	{
		\homo_1(\lev(z,|\K|)):
		&\VecV_{s_0} \ar[r] \ar@{=}[d]
		& \VecV_{a_1}  \ar@{=}[d]
		& \VecV_{s_1} \ar@{=}[d]\ar[l]
		&\cdots \ar[r] 
		%& \homo_0(\K_1^2) \cdots\ar@{->}[r] \ar@{=}[d]
		& \VecV_{a_m}\ar@{=}[d]
		& \VecV_{s_m} \ar[l]\ar@{=}[d]
		&
		\\
		\tilde\homo^0(\lev(z,R)):
		&\VecU^*_{s_0} \ar[r]
		& \VecU^*_{a_1} 
		& \VecU^*_{s_1} \ar[l]
		& \cdots \ar[r]
		%& \homo_0(\R_1^2) \cdots\ar@{->}[r]
		& \VecU^*_{a_m}
		& \VecU^*_{s_m}\ar[l]
		&
	}
	$$
The above diagram commutes and thus $\homo_1(\lev(z,|\K|))\cong \tilde\homo^0(\lev(z,R))$. By duality, $\tilde\homo^0(\lev(z,R))\cong\tilde\homo_0(\lev(z,R))$ establishing the claim.	
\end{proof}
\begin{figure}[ht!]
	\begin{center}
		\includegraphics[width=0.95\textwidth]{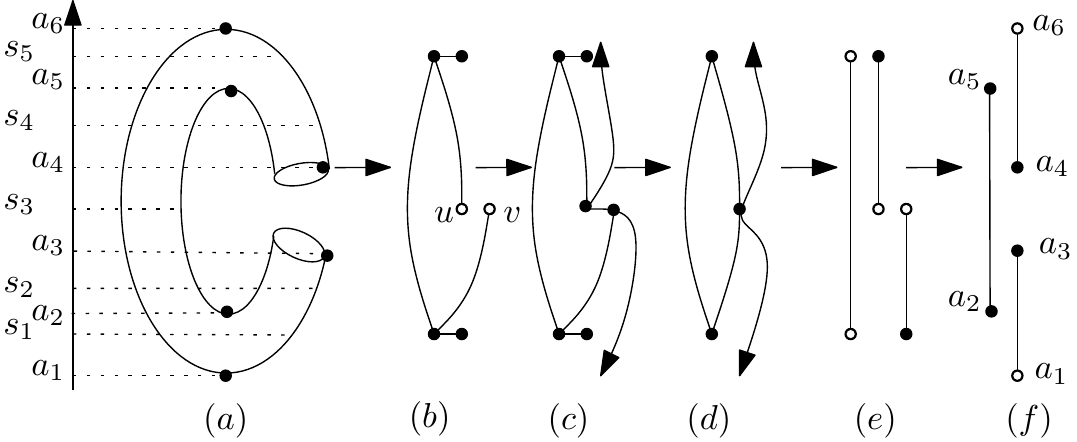}
	\end{center}
	\caption{(a) Illustration of stages of barcode extraction on a torus with a cylinder taken out: six critical values $a_1$-$a_6$ for $\homo_1$; (b) barcode graph; top two and bottom two vertices have the same values $s_5$ and $s_1$ respectively because of which they are drawn at the same levels; (c) threading connects the two open vertices at level $s_3$; (d) all vertices in the same connected component of a  level are coalesced (here $s_5, s_3, s_1$); (e) extracted bar code; (f) modifying the bars to bring their ends at the critical levels while reversing their types.}
	\label{barcode-graph}
\end{figure}
\subsection{Extracting bars}
\label{barcode-sec}
We apply the procedure of Agarwal et al.~\cite{AEHW06} for extracting the bars out of a barcode graph that they compute for surfaces without boundary in $\mathbb{R}^3$. Using the mergeable tree data structure of~\cite{tarjan}, this algorithm runs in $O(n\log n)$ time where $R$ has a total of $n$ edges and vertices. This algorithm in a sense mimics the definition of persistence
pairs in different diagrams based on their types as elucidated in~\cite{BCE12}.

Figure~\ref{barcode-graph} shows
the sequence of operations applied to the barcode graph of a torus with a cylinder taken out to extract the bars. Applying the barcode extraction algorithm of~\cite{AEHW06} straightforwardly on the barcode graph provides a wrong answer if threading is not done. The barcode graph (Figure~\ref{barcode-graph}(b)) 
is threaded first which may create additional cycles (Figure~\ref{barcode-graph}(c)).
Next, all vertices in a connected component of a level are contracted to a single vertex (Figure~\ref{barcode-graph}(d)). We apply the 
algorithm of~\cite{AEHW06} on this graph whose output are shown in Figure~\ref{barcode-graph}(e).

%\parpic[r]{\includegraphics[height=4cm]{barcode-extract.pdf}}

\vspace{0.1in}
\noindent
{\it Modifying the bars:} The algorithm of ~\cite{AEHW06} can be viewed as successively peeling off paths from the graph $R$ with endpoints at values $s_i$, $i=1,\ldots,m-1$ since $R$ has vertices at these levels only. An endpoint of a bar which is peeled from a split or a merge vertex $v$ is necessarily open (see bottom  and top vertices of the bar between $s_1$ and $s_5$ in Figure~\ref{barcode-graph}(e)). The last copy of $v$ that remains after all bars are peeled off of it remains to be closed (top and bottom vertices of the two shorter bars in Figure~\ref{barcode-graph}(e)). This accounts for the fact that a vertex contributes only a single component at its level. The other endpoints that arise from degree-1 vertices become open or closed according to the type of the vertex. After extracting all bars, we need to move their endpoints to the values $a_i$, $i=1,\ldots,m$ because the actual bars for $\K$ have endpoints at the critical values. 

Analogous to converting interval modules to bring their endpoints at critical values, we deploy the following conversions:
(i) $(s_i,s_j)\Rightarrow [a_{i+1},a_{j-1}]$, (ii) $[s_i,s_j)\Rightarrow (a_{i-1},a_{j-1}]$, (iii) $(s_i,s_j]\Rightarrow [a_{i+1},a_{j+1})$, (iv) $[s_i,s_j]\Rightarrow (a_{i-1},a_{j+1})$ where open and closed brackets indicate the open and closed ends respectively. See Figure~\ref{barcode-graph}(f).

\section{Reeb graph, barcode, and generators}
%Recall that the homology group $\homo_i(\K)$ is the direct
%sum of infinite interval modules or infinite 
%barcodes of the form $[a_i,\infty)$
%in the sublevel set persistence. Each of these infinite barcodes
%in turn has a one-to-one
%correspondence with an open-open barcode
%$(a_j,a_i)$ for $\homo_{i-1}$ or a closed-closed barcode $[a_i,a_j]$
%for $\homo_i$ in the level set persistence, see Theorem~\ref{}. 
Recall that $\homo_p(|\K|)\cong\opbarhom_{p-1}(z,|\K|)\oplus\clbarhom_p(z,|\K|)$
where $\opbarhom_{p-1}(z,|\K|)$ is generated by the 
open-open bars in $\homo_{p-1}(\lev(z,|\K|))$
and $\clbarhom_p(z,|\K|)$ is generated by the  
closed-closed bars in $\homo_p(\lev(z,|\K|))$~\cite{BD13}.
Our algorithm in the previous section produces the bars
for the level set module $\homo_1(\lev(z,|\K|))$ which allows us
to obtain the closed-open and closed-closed bars for 
the sublevel set module (standard persistence) $\homo_1(\sublev(z,|\K|))$
and the open-open bars for $\homo_2(\sublev(z,|\K|))$.
Although this completes the barcode for the second homology $\homo_2$,
we still need to compute the open-open bars for $\homo_0(\lev(z,|\K|))$
to complete the barcode for $\homo_1(\sublev(z,|\K|))$.
We achieve this with the help of Reeb graphs.

\vspace{0.1in}
\noindent
{\bf Reeb graphs and barcodes.}
Given a continuous function $f:|\K|\rightarrow \mathbb{R}$, one 
defines the Reeb graph $\reeb_f(|\K|)$ as the quotient space $|\K|\!\sim$
under the equivalence relation $\sim$  where for any
pair $x,y\in |\K|\times|\K|$, $x\sim y$ if and only if $f(x)=f(y)$ and
the level set $f^{-1}(f(x))$ contains $x$ and $y$ in the same connected
component. We observe the following connection:

\begin{proposition}
	$\opbarhom_0(z,|\K|)\cong \opbarhom_0(z,\reeb_z(|\K|))\cong \homo_1(\reeb_z(|\K|))$.
	\label{reeb-connection}
\end{proposition}
\begin{proof}
	By Theorem~\ref{main-thm}, $\homo_1(\reeb_z(|\K|))$ is isomorphic to the 
	direct sum
	$\opbarhom_0(z,\reeb_z(|\K|))\oplus\clbarhom_1(z,\reeb_z(|\K|))$.
	Consider an
	embedding that takes each vertex $v$ of $\reeb_z(|\K|)$ to points
	in $\mathbb{R}^3$ with the $z$-coordinate equaling the $z$-value
	of its pre-image in $|\K|$.
	Recall that an embedding necessarily maps distinct vertices to distinct
	points even if they have same $z$-values.
	An edge $(v_i,v_j)$ is embedded as the line segment joining $v_i$ and $v_j$.
	Observe that, because of the assumption that the height function
	$z$ is proper for $|\K|$, the embedded Reeb graph, also
	denoted $\reeb_z(|\K|)$ for simplicity, does not have any edge connecting
	two vertices with the same $z$-value. 
	So, no edge lies entirely on any level set $z^{-1}(r)$ for any $r\in\mathbb{R}$.
	Assuming general position, no two edges cross. 
	
	A level set $z^{-1}(r)$ for $\reeb_z(|\K|)$ has only isolated
	points where the edges intersect the plane $\pi_r:z=r$ transversely. 
	Thus, there is no closed-closed bar in $\homo_1(\lev(z,\reeb_z(|\K|))$
	and hence the summand group $\clbarhom_1(z,\reeb_z(|\K|))$ is trivial.
	It follows that 
	\begin{equation}
	\homo_1(\reeb_z(|\K|))\cong \opbarhom_0(z,\reeb_z(|\K|)).
	\label{reeb-eq}
	\end{equation}
	
	It follows from the definition of the Reeb graph that the $0$-dimensional
	homology group of any level set $z^{-1}(r)$, $r\in \mathbb{R}$,
	for $|\K|$ and $\reeb_z(|\K|)$ are isomorphic. Because of tameness
	of $z$, the same can be concluded for
	corresponding interval sets. Therefore, using our
	notation for interval sets between consecutive
	non-critical values, $\K_i^j=|\K|_{[s_i,s_j]}$ and
	$\R_i^j=\reeb_z(|\K|)_{[s_i,s_j]}$, we have the
	following commutative diagram between the $0$-dimensional level
	set zigzag persistence modules. 
	
	$$
	\xymatrix@C=3ex
	{
		\homo_0(\lev(z,|\K|)):
		&\homo_0(\K_0^0) \ar[r] \ar@{=}[d]
		& \homo_0(\K_0^1)  \ar@{=}[d]
		& \homo_0(\K_1^1) \ar[l] \cdots\ar[r] \ar@{=}[d]
		%& \homo_0(\K_1^2) \cdots\ar@{->}[r] \ar@{=}[d]
		& \homo_0(\K_{m-1}^m)\ar@{=}[d]
		& \homo_0(\K_m^m) \ar[l]\ar@{=}[d]
		&
		\\
		\homo_0(\lev(z,\reeb_z(|\K|))):
		&\homo_0(\R_0^0) \ar[r]
		& \homo_0(\R_0^1) 
		& \homo_0(\R_1^1) \ar[l] \cdots\ar[r]
		%& \homo_0(\R_1^2) \cdots\ar@{->}[r]
		& \homo_0(\R_{m-1}^m)
		& \homo_0(\R_m^m) \ar[l]
		&
	}
	$$
	
	It follows that the two zigzag persistence modules $\homo_0(\lev(z,|\K|))$
	and $\homo_0(\lev(z,\reeb_z(|\K|)))$ are isomorphic. Therefore,
	$\opbarhom_0(z,|\K|)\cong\opbarhom_0((z,\reeb_z(|\K|))$.
	Combining it with~\ref{reeb-eq} we get the claim. 
\end{proof}

\vspace{0.1in}
\noindent
{\bf Computing $\homo_1$-generators.}
Since $\homo_1(|\K|)\cong\opbarhom_{0}(z,|\K|)\oplus\clbarhom_1(z,|\K|)$,
we are required to generate a set of cycles whose classes 
form a basis for $\clbarhom_1(z,|\K|)$ and another set
of cycles whose classes form a basis for $\opbarhom_0(z,|\K|)$.
A closed-closed bar $[a_i,a_j]$ in $\homo_1(\lev(z,|\K|))$
is initiated by a cycle $C$ at the level $z(v_i)=a_i$ for some vertex
$v_i\in \K$. The homology class $[C]$ can be traced at each level set in the 
interval $[a_i,a_j]$ through the images and inverse images
of the inclusion maps that produce the zigzag level set
module $\homo_1(\lev(z,|\K|))$. In other words, the class 
$[C]\in \homo_1(|\K|)$ represents
the bar $[a_i,a_j]$. Therefore, the classes of cycles
initiating closed-closed bar in the level set persistence module
$\homo_1(\lev(z,|\K|))$ generate the summand $\clbarhom_1(z,|\K|)$ of $\homo_1(|\K|)$.

We compute the cycles initiating the closed-closed bars as 
follows. A bar with a closed end results from a split that occurs
during updating $\overrightarrow{G}_{s_{i-1}}$ to $\overrightarrow{G}_{a_i}$. 
A new primary cycle $C$ is born in both cases of the split,
which in turn initiates a new edge, say $e$ in the barcode graph $R$. 
We can keep the cycle $C$ associated with $e$ in $R$.
After we extract all bars
from the barcode graph, we can determine the closed-closed bars
and determine the cycles associated with their initiating edges. 
The drawback of this approach
is that we may store many unnecessary cycles that initiate closed-open
bars. To avoid this, we do not store the entire cycle beforehand
initiating a bar from a split. Instead, we store one directed
edge $d$ in the cycle $C$ associated with the edge $e\in R$. We also
remember the vertex of $\K$ where the split has occurred. 
After we extract a closed-closed
bar from the barcode graph, we obtain its associated directed edge $d$
and the vertex $v\in \K$, and trace out the cycle containing $d$ in 
the level set graph $\overrightarrow{G}_{z(v)}$. Taking into account the time to create $R$
and the time to extract the cycles, this process cannot take 
more than $O(n\log n +k)$ time for all cycles to be output 
where $k$ is their total size.

To compute the generating cycles
for the summand $\opbarhom_0(z,|\K|)$ of $\homo_1(|\K|)$, we 
use the second equivalence in Proposition~\ref{reeb-connection}.
It implies that if a cycle basis for the Reeb graph is mapped injectively to
a sub-basis of $\homo_1(|\K|)$, then that sub-basis indeed generate
the summand $\opbarhom_0(z,|\K|)$. 

A cycle basis for $\homo_1(\reeb_z(|\K|))$ can be computed in linear time
by computing a spanning tree of $\reeb_z(|\K|)$ treating it as a graph and then generating a cycle
for each additional edge not in the spanning tree. The pre-image of these
basis cycles w.r.t. the surjective map $\phi: |\K|\rightarrow \reeb_z(|\K|)$
can be computed again in time linear in the total size of the basis cycles
and the Reeb graph. These pre-images can also be deformed with a homotopy
to the $1$-skeleton
of $\K$. The pre-images thus constructed form a cycle basis
of the so called {\em vertical homology} group which is
a summand group of $\homo_1(|\K|)$~\cite{DW13}. Therefore, the pre-images
form a cycle basis of $\opbarhom_0(z,|\K|)$ which
can be computed in $O(n\log n + k)$ time where $k$ is the total size of all 
such cycles. 

\vspace{0.1in}
\noindent
{\bf Computing $\homo_2$-generators.}
We know $\homo_2(|\K|)\cong \opbarhom_1(z,|\K|)\oplus\clbarhom_2(z,|\K|)$.
Since $|\K|\subset\mathbb{R}^3$, $\clbarhom_2(z,|\K|)$ is trivial
because there are no $2$-cycles in any level set. Therefore, a set of 
independent $2$-cycles in $|\K|$ that map bijectively to 
a basis of $\opbarhom_1(z,|\K|)$ form
a set of basis cycles for $\homo_2(|\K|)$.

Let $C$ be any primary cycle that initiates an open-open bar
extracted from the barcode graph $R$. As we already explained,
the component $Z$ of $R$ providing the open-open
bar in this case corresponds to a $2$-cycle in $|\K|$.
We can think $R$ as a $1$-complex linearly embedded
in $\mathbb{R}^3$ with a vertex $v\in R$ having height $z(v)$.
Then, there is a continuous surjective map $\xi: E\rightarrow R$
where $E\subseteq \mathbb{R}^3\setminus |\K|$ 
is the union of all faces bounded by
primary cycles over all levels $r\in\mathbb{R}$. In fact, $R$ is the
Reeb graph of the height function $z:E\rightarrow\mathbb{R}$.  An open-open
bar in $\homo_1(\lev(z,|\K|))$ signifies a non-trivial
class of $\homo_2(|\K|)$ by Theorem~\ref{main-thm}. 
The boundary of the inverse image $\xi^{-1}(Z)$ is a $2$-cycle in $\K$,
and it can be argued that it is independent of all such cycles.
We can extract this $2$-cycle by taking any complex triangle $t\in \K$ where
$d\cdot\tri=t$ for a directed edge $d\in C$, and then collecting all triangles
that bound the void whose boundary includes $t$.
This can be done by a simple depth-first walk in the adjacency data structure
of $\K$. In total, after $O(n)$-time walk, we collect all $2$-cycles
generating a basis for $\homo_2(|\K|)$.

\begin{newremark}
	We observe that all of the computations
	that we described can be adapted to the case
	when $\K$ is not necessarily pure.
	During the level set updates, 
	we do not let those cycles $C_{\overrightarrow{F}}$
	generate nodes and edges in the barcode graph $R$ 
	where the face $F$ (on right)
	lies on the intersection of the level set with tetrahedra.
	This is because $C_F$ cannot be a cycle basis element of the
	level set graph containing $C_F$.
	The edges and vertices of $\K$ that are 
	not adjacent to any triangle do not affect the
	computation of the level set zigzag persistence because
	they do not contribute to any
	primary cycle. However, they appear in the
	Reeb graph computation and may contribute to the
	open-open bars and hence infinite bars for
	the first homology.
\end{newremark}
\begin{newremark}
	Because of Proposition~\ref{secondary},
	if we run our level set update algorithm with 
	the roles of primary and secondary cycles switched, we obtain the
	Reeb graph $\reeb_z(|\K|)$ as the barcode graph. This gives
	an alternate $O(n\log n)$-time algorithm for computing Reeb graphs.
\end{newremark}

\vspace{0.1in}
\noindent
{\it Proof of Theorem~\ref{main1-thm}}: We have already presented an $O(n\log n)$ time algorithm for computing $\homo_1(\lev(z,|\K|))$. Notice that since level sets reside in planes, $\homo_i(\lev(z,|\K|))$ is trivial for every $i>1$. To compute $\homo_0(\lev(z,|\K|))$, we need to compute the connected components for level sets and track them. The barcode graph in this case is exactly the Reeb graph which can be computed in $O(n\log n)$ time by the algorithm of Parsa~\cite{Parsa}. The barcode from this graph can be extracted as before.
The basis cycles for $\homo_1(\K)$ and $\homo_2(\K)$ can be computed in $O(n\log n +k)$ time as described before. The case for $\homo_0(\K)$ is trivial because connected components of $\K$ can be determined in linear time by a depth first serach in the $1$-skeleton of $\K$.

\vspace{0.1in}
\noindent
{\it Proof of Theorem~\ref{main2-thm}}: Because of Theorem~\ref{main-thm}, all bars of $\homo_1(\sublev(z,|\K|))$ except the infinite bars that correspond to the open-open bars of $\homo_0(\lev(z,|\K|))$ can be obtained by computing $\homo_1(\lev(z,|\K|))$. We can use the first equivalence in Proposition~\ref{reeb-connection}
to derive the infinite bars of $\homo_1(\sublev(z,|\K|))$
that correspond to the open-open bars of $\homo_0(\lev(z,|\K|))$.
We compute the Reeb graph $\reeb_z(|\K|)$ in $O(n\log n)$ time using the
algorithm by Parsa~\cite{Parsa} and then extract the open-open bars
of $\homo_0(z,\reeb_z(|\K|))$ by running the $O(n\log n)$-time extended
persistence algorithm of Agarwal et al.~\cite{AEHW06} on it.

Since the level sets of $z$ on $|\K|$ resides on planes, $\homo_2(\lev(z,|\K|))$ is trivial.
Hence, there is no finite bars in the standard persistence $\homo_2(\sublev(z,|\K|))$. Also, for the same reason, the infinite bars of $\homo_2(\sublev(z,|\K|))$ correspond only to the open-open bars of $\homo_1(\lev(z,|\K|))$ which can be computed in $O(n\log n)$ time as we have described already. For $\homo_0(\sublev(z,|\K|))$, we can compute the closed-open bars and closed-closed bars of $\homo_0(\lev(z,|\K|))$ to obtain the finite and infinte bars respectively. We have already mentioned that this can be done by computing the Reeb graph and extracting bars from it in $O(n\log n)$ time.

\section{Discussions}
This work has spawned some interesting questions. 
The foremost among them is perhaps
the question of being able to extend the presented approach toward computing the
general persistence for simplicial complexes $\K$ embedded in $\mathbb{R}^3$ while maintaining an $O(n\log n)$ time complexity. In this case, the level sets are embedded on surfaces of possibly high genus that can themselves change topology as sweep proceeds. It is not clear how to track a basis efficiently in this case. 
One possibility is to look for functions that can be transformed into
height or height-like functions. Our approach applies to
functions on $|\K|$ that can be extended continuously
to entire $\mathbb{R}^3$ with level sets being points, planes, or spheres. The height
function on $|\K|$ is one such function whose extension to entire $\mathbb{R}^3$ has such level sets which are merely planes.

%What about extending our approach to complexes in higher dimensions, say $2$-complexes in $\mathbb{R}^4$? There are two roadblocks. First, the level sets are now general graphs embedded in $\mathbb{R}^3$. It is not clear how to compute the barcode graph efficiently in this context. Second, we can no more take care of the closed nodes of the barcode graphs with the thread representing the unbounded component since, in $\mathbb{R}^3$, $\homo_1(\cdot)$ is no more isomorphic to $\tilde\homo_0(\cdot)$ of the complement.

With our approach we can compute the homology generators in 
$O(n\log n +k)$ time. With these
generators, is it possible to compute the co-homology generators as well
efficiently? This will allow annotating the simplices in the sense
of~\cite{BCCDW12} so that the homology class of any given cycle can be
determined efficiently. One can compute a cohomology basis or an annotation
from scratch in matrix multiplication time~\cite{BCCDW12}, 
but can we leverage
the fact that a homology basis is already available? 

We have exploited the embedding of a complex to compute the persistence
efficiently for a special class of PL functions. Can we do the same for 
some interesting class of filtrations; or, even for a sequence of embedded
complexes connected with simplicial maps? Finally, can our approach be extended to dimensions beyond $\mathbb{R}^3$ which will imply breaking through the current matrix multiplication time barrier for computing persistence in general?
%After finalizing this paper, the author has discovered recently that the technique indeed can be extended to general filtrations albeit with some non-trivial observations while worsening the time complexity to $O(n^2\log n)$~\cite{Dey18}. This new paper also explains the relation of the
%barcode graph to the level set persistence in a more general context.

\vspace{0.1in}
\noindent
{\bf Acknowledgment.} The author acknowledges the support of the NSF grants CCF-1740761, CCF-1526513, and DMS-1547357 for this research.

\end{document}